\documentclass{amsart}\usepackage{amssymb, amsmath, bbold}
\usepackage{pstricks}
\newtheorem{thm}{Theorem}[section]
\newtheorem{lemma}[thm]{Lemma}

\newtheorem{cor}[thm]{Corollary}

\theoremstyle{definition}

\def \mco{{\mathcal O}}
\def \mcr{{\mathcal R}}
\def \Real {{\mathbb R}}
\def \Sphere {{\mathbb S}}
\def \Vol {\operatorname{Vol}}
\def \Tr {\operatorname{Tr}}
\def \comp {\operatorname{comp}}
\def \loc {\operatorname{loc}}
\def \rank {\operatorname{rank}}
\def \Complex {\mathbb{C}}

\def \bbone {\mathbb{1} }
\title [Resonances and balls ]
{Resonances and  balls in obstacle scattering 
with Neumann boundary conditions }
   \author { T. J. Christiansen}
\thanks{Partially supported by NSF grant DMS 0500267.}
\begin{document}
\begin{abstract}
We consider scattering by an obstacle in $\Real^d$, $d\geq 3 $ odd.  
We show that for the Neumann Laplacian if an obstacle has the same 
resonances as the ball of radius $\rho$ does, then the obstacle is a ball of
radius $\rho$.  We give related results for obstacles which are disjoint unions
of several balls of the same radius.
 \end{abstract}

\maketitle

\section{Introduction}
The purpose of this is note is to show that for obstacle scattering in 
$\Real^d$, $d\geq 3$ odd, a ball is uniquely determined by its resonances
for the Laplacian with Neumann boundary conditions.  We actually show a 
somewhat stronger result: If $\mco_1$ and $\mco_2$ have the same 
(Neumann) resonances and $\mco_1$ is the disjoint union of $m$ balls,
each of radius $\rho$, then so is $\mco_2.$

We fix some notation.  Let $\mco \subset \Real^d$ be a compact 
 set of dimension $d$ with smooth
boundary $\partial \mco$.  Let $\Delta_{\Real^d\setminus
\mco}$ denote the Laplacian with {\em Neumann} boundary conditions on
$\Real^d \setminus \mco$.  Set $R_{\mco}(\lambda)=(\Delta_{\Real^d\setminus
\mco}-\lambda^2)^{-1}$, with the convention that $R_{\mco}(\lambda)$ 
is bounded on $L^2(\Real^d \setminus \mco)$ when $\Im \lambda >0$.
Then, for odd $d$, it is known that $R_{\mco}(\lambda):
L^2_{\comp}(\Real^d \setminus \mco) \rightarrow 
L^2_{\loc}(\Real^d \setminus \mco)$ has a meromorphic continuation to 
$\Complex$.  Set
$$\mcr_{\mco}=\{ \lambda_0\in \Complex: \; R_{\mco}(\lambda) \;
\text{has a pole at }\; \lambda_0,\; \text{repeated according to 
multiplicity}\}.$$
The multiplicity can be defined via 
$$m_{\mco}(\lambda)=\rank \oint_{|z-\lambda|=\epsilon} R_{\mco}(z)dz,\;
0<\epsilon \ll 1.$$
Set $B(\rho)$ to be the closed ball of radius $\rho$ centered at the origin.

\begin{thm}\label{thm:ball}
If $d\geq 3$ is odd, and $\mco \subset \Real^d$ is a smooth compact 
set with $\mcr_{\mco}=\mcr_{B(\rho)}$, then $\mco$ is a translate of $B(\rho)$.
\end{thm}
We note that a simple scaling argument shows that if $\rho_1\not = \rho_2$
then $\mcr_{B(\rho_1)} \not = \mcr_{B(\rho_2)} $ so that $B(\rho)$ is 
determined (up to translation) by its Neumann resonances.

Hassell and Zworski \cite{h-zw} proved that for the {\em Dirichlet}
Laplacian in $\Real^3$, if a connected set $\mco$ has the same 
resonances as $B(\rho)$, then $\mco$ is a translate of $B(\rho)$. It
seems that their argument also works for Neumann boundary 
conditions, again in dimension $3$ with the restriction that $\mco$ is 
connected. That we
can prove this result for all odd dimensions
 for Neumann boundary conditions follows
partly from 
the fact that for the Neumann case the resonances determine the determinant
of the scattering matrix up to one unknown constant, while for the Dirichlet
case there are two unknown constants-- compare Lemma \ref{l:rsp} to 
\cite{h-zw}.

We shall actually prove the 
following theorem, from which Theorem \ref{thm:ball} follows immediately.
\begin{thm}\label{thm:mballs} Let $d\geq 3$ be odd.
Let $\mco_1\subset \Real^d$ 
be a disjoint union of $m$ closed balls each of the same radius $\rho$,
and let $\mco_2\subset \Real^d$ be a smooth compact set.
If $\mcr_{\mco_1}=\mcr_{\mco_2}$, then $\mco_2$ is also a 
disjoint union of $m$ closed balls of radius $\rho$.
\end{thm}

The proof of Theorem \ref{thm:mballs} uses heat coefficients.
These are closely related to the singularity at $t=0$ 
of the distribution
\begin{equation}\label{eq:wavetrace}
u_{\mco}(t)=\Tr \left( \cos\left(t\sqrt{\Delta_{\Real^d\setminus \mco}}\right)
-\cos\left(t\sqrt{\Delta_{\Real^d}}\right)\right).
\end{equation}
This is rather informal.  One way to make 
precise sense of (\ref{eq:wavetrace}) is 
as follows.  With $\rho'$ chosen sufficiently
large that $\mco \subset B(\rho')$,
\begin{multline}\label{eq:wavetrace2}
u_{\mco}(t)
= 
\Tr \left( \cos\left(t\sqrt{\Delta_{\Real^d\setminus \mco}}\right)
-\bbone_{\Real^d\setminus B(\rho')}
\cos\left(t\sqrt{\Delta_{\Real^d}}\right)\bbone_{\Real^d\setminus B(\rho')}
\right) + \Tr\left( 
\bbone_{B(\rho')} \cos\left(t\sqrt{\Delta_{\Real^d}}\right)
\bbone_{ B(\rho')}\right)
\end{multline}
where $\bbone_{E}$ is the characteristic function of the set $E$, 
compare \cite{zwpfe}.
The Poisson formula for resonances in odd dimensions is
$$u_{\mco}(t)=\sum_{\lambda_j\in \mcr_{\mco}}e^{i\lambda_j|t|},\; t\not = 0$$
 see 
\cite{b-g-r,sttwg,pbsp, zwpf}, 
and \cite{sj-zw2} for an application to the existence of resonances.
This shows that 
any singularities of the distribution $u_{\mco}(t)$
at nonzero time are determined
by the resonances of $\Delta_{\Real^d \setminus \mco}$.
When $\mco$ is not connected, we expect that the 
distribution of (\ref{eq:wavetrace2})
has singularities at nonzero times, see \cite{p-st}.  In particular,
 if $\mco$ consists of two
disjoint convex obstacles, the distribution (\ref{eq:wavetrace2}) has a singularity at 
twice the distance between them (e.g. \cite{g-m} or \cite[Section 6.4]{p-st}). 
 From this and Theorem \ref{thm:mballs}
we have
\begin{cor}Let $d\geq 3$ be odd, and let 
$\mco_1,\; \mco_2\subset \Real^d$ be smooth compact sets.
Suppose $\mco_1 $ is the disjoint union of $2$ balls of radius $\rho$ a
distance $\delta>0$ apart.  If $\mcr_{\mco_1}=\mcr_{\mco_2}$, then $\mco_2$ is 
obtained from $\mco_1$ by a rigid motion.
\end{cor}

This paper was inspired by \cite{h-zw}, and we use some of the same 
notation. 

\section {Proof of Theorem \ref{thm:mballs}}
Let $S_{\mco}(\lambda)$
be the scattering matrix for the {\em Neumann} Laplacian on $\Real^d
\setminus \mco$, and set
$s_{\mco}(\lambda)=\det S_{\mco}(\lambda).$
Set $$E(z)=(1-z)\exp\left(\sum_1^d \frac{z^j}{j} \right).$$

In the lemma below we have not made any assumptions on 
the connectivity of $\Real^d \setminus \mco.$  If $\Real^d\setminus
\mco$ has a bounded component, then $\mcr_{\mco}$ contains 
the square roots, and the negative square roots, 
of eigenvalues of the Neumann
Laplacian on the bounded component(s).  These do not cause poles of the
determinant of the scattering matrix.  
We remark that it is possible to have $0\in \mcr_{\mco}$ here, if 
$\Real^d\setminus \mco$ has a bounded component, but not if $\Real^d\setminus
\mco$ is connected.
\begin{lemma}\label{l:rsp}
Let $d\geq 3$ be odd and $\mco \subset \Real^d$ be a smooth compact set.
Then 
$$s_{\mco}(\lambda)=e^{ic_{\mco}\lambda^d}\prod_{\lambda_j\in \mcr_{\mco},
\; \lambda_j \not = 0}
\frac{E(-\lambda/\lambda_j)}{E(\lambda/\lambda_j)}$$
for some real constant $c_{\mco}$.
\end{lemma} 
For 
a large number of more 
 general
self-adjoint operators (for example, for
the Laplacian with different boundary conditions),
the resonances determine the determinant of the scattering matrix up to two
unknown constants, e.g. \cite{h-zw}.  
\begin{proof}It follows from \cite{zwpf} that
$$s_{\mco}(\lambda)=e^{i g_{\mco}(\lambda)}\prod_{\lambda_j\in \mcr_{\mco}, \; \lambda_j \not =0}
\frac{E(-\lambda/\lambda_j)}{E(\lambda/\lambda_j)}$$
where $g_{\mco}(\lambda)$ is a polynomial of degree at most $d$.
It remains only to show that $g_{\mco}(\lambda)= c_{\mco}\lambda^d$.

We now use the expression of \cite{p-zw}
for the scattering matrix.
Choose $\rho'>0$ so that $\mco$ is contained in the ball of radius $\rho'$
 centered
at the origin.  For $\psi\in C_c^{\infty}(\Real^d)$,
let 
$${\mathbb E}_{\pm}^{\psi}(\lambda):L^2(\Real^d)\rightarrow
 L^2(\Sphere^{d-1})$$
be the operator with Schwartz kernel $\psi(x)\exp(\pm i \lambda \langle x, 
\omega \rangle )$.  For $i=1,\; 2,\; 3$, choose $\chi_i \in C^{\infty}_c
(\Real^d)$ so that $\chi_i\equiv 1$ if $|x|<4+\rho'$
 and $\chi_{i+1}\equiv
1$ on the support of $\chi_i$, $i=1,\; 2.$    Set
\begin{equation}\label{eq:A}
A_{\mco}(\lambda)= \tilde{c}_d\lambda^{d-2}{\mathbb E}^{\chi_3}_+(\lambda)
[\Delta_{\Real^d}, 
\chi_1]
R_{\mco}(\lambda)[\Delta_{\Real^d},\chi_2]  ^t{\mathbb E}_-^{\chi_3}(\lambda).
\end{equation}
Here $\tilde{c}_d=i\pi(2\pi)^{-d}$
and $^t{\mathbb E}_-^{\chi_3}$ denotes the transpose of 
${\mathbb E}_-^{\chi_3}$.
Then the scattering matrix $S_{\mco}(\lambda)$ associated to $P$ is given by 
\begin{equation}\label{eq:S}
S_{\mco}(\lambda)=I+A_{\mco}(\lambda).
\end{equation}

Next, note that for the Neumann Laplacian, 
$$
R_{\mco}(0)[\Delta_{\Real^d},\chi_2]  ^t{\mathbb E}_-^{\chi_3}(0)= \chi_2$$
since $[\Delta_{\Real^d},\chi_2]^t{\mathbb E}_-^{\chi_3}(0)=
\Delta_{\Real^d} \chi_2$, and 
$\chi_2$ satisfies the Neumann boundary conditions.  But then 
$${\mathbb E}^{\chi_3}_+(0)[\Delta_{\Real^d}, 
\chi_1]\chi_2 = \int_{\Real^d} \Delta_{\Real^d} \chi_1  =0$$
so that $A_{\mco}(\lambda)$ vanishes to order at least $d-1$ at $\lambda=0$.
Thus
 \begin{equation}
s_{\mco}(\lambda)-1=\det(I+A_{\mco}(\lambda))-1=O(|\lambda|^{d-1})
\end{equation}
at $\lambda =0$.  Using the fact that $E(\lambda/\lambda^j)=O(|\lambda|^{d})
$ at $\lambda=0$, we see that $g^{(j)}(\lambda)=0$ for $j=0,...,d-2$.
It is also well known that for $\lambda \in \Real$, 
$s_{\mco}(\lambda)s_{\mco}(-\lambda)=1$, implying that $g^{(d-1)}(0)=0$.
Finally, since for $\lambda \in \Real$ $\overline{s_{\mco}(\lambda)}
= s_{\mco}(-\lambda)$, we have that $g^{(d)}(0)$ is real.  This finishes the 
proof.
\end{proof}

We fix some notation.  For a smooth compact set $\mco \subset \Real^d$ with 
boundary $\partial \mco$ of dimension $d-1$, let $\kappa_{1,
\partial \mco} \dots ,\kappa_{d-1, \partial \mco}$ denote the principal curvatures of 
$\partial \mco$.  We use the normalization that the mean curvature $H_{\partial \mco}$
is $
H_{\partial \mco}= \sum_1^{d-1}\kappa_{j,
\partial \mco}$.  
\begin{lemma}\label{l:heatinv} Let
$d\geq 3$ be odd and let $\mco\subset \Real^d$ be a smooth compact set with boundary $\partial \mco$
of dimension $d-1$.  Then $\mcr_{\mco}$ determines
$$
\Vol(\partial \mco), \; \int_{\partial \mco}H_{\partial \mco},\; \text{ and }\; 
\int _{\partial \mco}\left(13 H^2_{\partial \mco}
+ 2\sum_1^{d-1}\kappa^2_{j,\partial \mco}\right).$$
\end{lemma}
\begin{proof}
As $t\downarrow 0$,
\begin{equation}\label{eq:he}
\Tr (e^{-t\Delta_{\Real^d \setminus \mco}}-e^{-t\Delta_{\Real^d}})
\simeq t^{-d/2}\sum_{n=0}^{\infty} t^{n/2} a_n.
\end{equation}
The trace in (\ref{eq:he}) can be made precise in exactly the 
same manner that (\ref{eq:wavetrace}) is made precise in (\ref{eq:wavetrace2}).
Explicit expressions for the first few $a_n$ can be 
found in  \cite{br-gi}.  For Neumann boundary conditions in $\Real^d$, we have 
\begin{align*} 
a_0& =\alpha_0 \Vol(\mco)\\
a_1&= \alpha_1 \Vol(\partial \mco)\\
a_2& = \alpha_2 \int _{\partial \mco} H_{\partial \mco}\\
a_3& = \alpha_3 \int _{\partial\mco}\left(13 H^2_{\partial \mco}
+ 2\sum_1^{d-1}\kappa^2_{j,\partial \mco}\right)
\end{align*}
where the $\alpha_i$ are nonzero constants which depend on the dimension.
Using Lemma \ref{l:rsp}, (\ref{eq:he}), 
$$\Tr (e^{-t\Delta_{\Real^d \setminus \mco}}-e^{-t\Delta_{\Real^d}})
= \frac{1}{2\pi i}\int_{-\infty}^{\infty} 
e^{-t\lambda^2}\frac{d}{d\lambda}\log s_{\mco}(\lambda)d\lambda
+\frac{1}{2}\sum _{\lambda_j \in \mcr_{\mco}\cap \Real} e^{-t\lambda_j^2}$$
(e.g. \cite[Corollary 2.10]{robert} for the case of $\Real^d\setminus \mco$ 
connected, or \cite{bb})
and the expressions for the $a_i$ above, we prove the lemma.
\end{proof}

We are now ready for the proof of our theorem.
\begin{proof}[Proof of Theorem \ref{thm:mballs}]
Suppose $\mco_1$, $\mco_2$ are as in the statement of the theorem,
with $\mcr_{\mco_1}=\mcr_{\mco_2}.$
 Then by Lemma \ref{l:heatinv}, 
\begin{multline}\label{eq:same}
\Vol(\partial \mco_1)=\Vol(\partial \mco_2),\;  
\int_{\partial \mco_1}H_{\partial \mco_1}= 
\int_{\partial \mco_2}H_{\partial \mco_2}
\;\;
\text{and}\\ \int _{\partial \mco_1}\left(13 H^2_{\partial \mco_1}
+ 2\sum_1^{d-1}\kappa^2_{j,\partial \mco_1}\right)=
\int _{\partial \mco_2}\left(13 H^2_{\partial \mco_2}
+ 2\sum_1^{d-1}\kappa^2_{j,\partial \mco_2}\right).
\end{multline}
By the Cauchy-Schwarz inequality, for $i=1,\; 2$, 
$$\frac{ \left( \int_{\partial \mco_i}H_{\partial \mco_i}\right)^2}
{\Vol(\partial \mco_i)} \leq \int_{\partial \mco_i}H^2_{\partial \mco_i}$$
with equality if and only if $H_{\partial \mco_i}$ is constant.
Likewise, $\sum_1^{d-1} \kappa^2_{j,\partial \mco_i}\geq \frac{1}{d-1}
\left( \sum_1^{d-1}\kappa_{j,\partial \mco_i}\right)^2$, with equality 
if and only if $\kappa_{1,\partial \mco_i}=\kappa_{2,\partial \mco_i}
=\cdot \cdot 
\cdot =\kappa_{d-1,\partial \mco_i}$.
Since equality holds for $\mco_1$, the three equalities of
(\ref{eq:same}) mean that $H_{\partial \mco_2}$
is a constant and the principal
curvatures of $\partial \mco_2$ are all equal.  Thus $\mco_2$ must be the
disjointn union of $l$ balls of fixed radius $\rho'$.  Again using the first 
two equalities of (\ref{eq:same}), we see that we must have $l=m$ and 
$\rho' =\rho$.    
\end{proof}

We remark that the proof of this theorem is a bit delicate, in the sense that 
it is important for us that
in the expression for $a_3$  the coefficients of 
$H^2_{\partial \mco}$ and $\sum_1^{d-1}\kappa^2_{j,\partial \mco}$ 
have the same 
sign.  However, if $\mcr_{\mco}=\mcr_{B(\rho)}$ 
{\em and} we know a priori that $\mco$ is a ($d$-dimensional) 
{\em convex} set, then we do not need to 
use the coefficient $a_3$ to prove that $\mco$ is, up to translation, $B(\rho)$.
Instead, one can use first two equalities
of (\ref{eq:same}) and the fact that in the Alexandrov-Fenchel inequality
$$\left( \frac{\Vol(\partial \mco)}{\Vol(\partial B(\rho))}\right)^{1/(d-1)}
\leq 
\left( \frac{\int_{\partial \mco}H_{\mco}}{ \int_{\partial B(\rho)}H_{B(\rho)}}\right)^{1/(d-2)} $$
equality holds if and only if $\mco$ is a ball 
\cite[Chapter 4, Section 9]{al} or
\cite{alorig}.

\small
\noindent
{\sc 
Department of Mathematics\\
University of Missouri\\
Columbia, Missouri 65211\\
e-mail: {\tt tjc@math.missouri.edu} }
\end{document}